\documentclass[a4paper]{llncs}
\usepackage{amssymb}
\usepackage{amsmath}
\usepackage{url}
\newcommand{\union}{\cup}
\newcommand{\inter}{\cap}
\renewcommand{\And}{\wedge}

\newcommand{\Union}{\bigcup}
\newcommand{\pow}[1]{{\sf pow}\ #1}

\newcommand{\emptylist}{[\,]}
\newcommand{\cons}[2]{#1\;\#\;#2} 

\newcommand{\setn}[1]{\{0, 1, \ldots, #1-1\}}        
\newcommand{\listn}[1]{[0, 1, \ldots, #1-1]}         
\newcommand{\combine}[2]{{\sf comb}\;#1\;#2} 
\newcommand{\permute}[1]{{\sf perm}\;#1} 
\newcommand{\families}[3]{{\sf fams}^{#1#2#3}} 

\newcommand{\card}[1]{|#1|}
\newcommand{\uc}[1]{{\sf uc}\ #1}
\newcommand{\uca}[2]{{\sf uc}_{#1}\ #2}
\newcommand{\closure}[1]{\left\langle#1\right\rangle}
\newcommand{\closurea}[2]{\left\langle#2\right\rangle_{#1}}
\newcommand{\ic}[2]{{\sf ic}\ #1\ #2}
\newcommand{\ica}[3]{{\sf ic}_{#1}\ #2\ #3}
\newcommand{\icatilde}[3]{{\sf \tilde{ic}}_{#1}\ #2\ #3}
\newcommand{\icaa}[4]{{\sf \tilde{ic}}_{#1}^{#4}\ #2\ #3}
\newcommand{\cnt}[2]{\#_{#1}#2} 
\newcommand{\frankl}[1]{{\sf frankl}\ #1}
\newcommand{\wf}[2]{{\sf wf}_{#2}\ #1} 
\newcommand{\sw}[2]{#1(#2)} 
\newcommand{\fw}[2]{#1(#2)} 
\renewcommand{\ss}[3]{\bar{#2}_{#3}(#1)} 
\newcommand{\fs}[3]{\bar{#2}_{#3}(#1)}  
\newcommand{\hc}[2]{{\sf hc}_{#1}^{#2}}  
\newcommand{\hs}[5]{\bar{#4}^{#2}_{#1#5}(#3)} 
\newcommand{\hcprj}[3]{{\sf hc}_{#1}^{#2}\left\lfloor{#3}\right\rfloor} 
\newcommand{\uce}[1]{{\sf uce}\ #1} 
\newcommand{\ssn}[2]{{\sf ssn}\ #1\ #2} 
\newcommand{\ssnaux}[5]{{\sf ssn}^{#3,#4,#5}\ #1\ #2} 
\newcommand{\ssnauxa}[4]{{\sf ssn}^{#3,#4}\ #1\ #2} 

\newcommand{\permuteset}[2]{{\sf perm\_set}\ #1\ #2}
\newcommand{\permutefamily}[2]{{\sf perm\_fam}\ #1\ #2}
\newcommand{\nefa}[3]{{\sf nef\_aux}^{#2}\ #1\ #3}
\newcommand{\nef}[2]{{\sf nef}^{#2}\ #1}

\newcommand{\filter}[2]{{\sf filter}\ #1\ #2}
\newcommand{\map}[2]{{\sf map}\ #1\ #2}
\newcommand{\foldl}[3]{{\sf foldl}\ #1\ #2\ #3}
\newcommand{\remdups}[1]{{\sf remdups}\ #1}
\newcommand{\sort}[1]{{\sf sort}\ #1}
\newcommand{\listsum}[1]{{\sf listsum}\ #1}
\newcommand{\nth}[2]{#1_{[#2]}}
\newcommand{\underrule}{\vbox{\hrule width.6em}}

\author{Filip Mari\' c, Miodrag \v Zivkovi\' c, Bojan Vu\v ckovi\' c}
\title{Formalizing Frankl's Conjecture: FC-families}
\institute{Faculty of Mathematics, University of Belgrade\thanks{The
    first author was partially supported by the Serbian Ministry of
    Education and Science grant 174021 and by the SNF grant SCOPES
    IZ73Z0127979/1, the second author by the Serbian Ministry of
    Education and Science grant 174021 and the third author by the
    Serbian Ministry of Education and Science grant 044006 (III).}}

\begin{document}
\maketitle

\begin{abstract}
  The Frankl's conjecture, formulated in 1979. and still open, states
  that in every family of sets closed for unions there is an element
  contained in at least half of the sets. FC-families are families for
  which it is proved that every union-closed family containing them
  satisfies the Frankl's condition (e.g., in every union-closed family
  that contains a one-element set {a}, the element a is contained in
  at least half of the sets, so families of the form {{a}} are the
  simplest FC-families). FC-families play an important role in
  attacking the Frankl's conjecture, since they enable significant
  search space pruning. We present a formalization of the computer
  assisted approach for proving that a family is an
  FC-family. Proof-by-computation paradigm is used and the proof
  assistant Isabelle/HOL is used both to check mathematical content,
  and to perform (verified) combinatorial searches on which the proofs
  rely. FC-families known in the literature are confirmed, and a new
  FC-family is discovered.
\end{abstract}

\section{Introduction}
\label{sec:introduction}
Formalized mathematics and interactive theorem provers (sometimes
referred to as proof assistants) have made great progress in recent
years. Many classical mathematical theorems have been formally proved
and proof assistants have been intensively used in hardware and
software verification. The most successful proof assistants now days
are Coq, Isabelle/HOL, HOL Light, etc.

Several of the most important results in formal theorem proving are
for the problems that require proofs with much computational
content. These proofs are usually highly complex (and therefore often
require justifications by formal means) since they combine classical
mathematical statements with complex computing machinery (usually
computer implementation of combinatorial algorithms). The
corresponding paradigm is sometimes referred to as
\emph{proof-by-evaluation} or \emph{proof-by-computation}. Probably,
the most famous examples of this approach are the proofs of the
Four-Color Theorem and the Kepler's conjecture.

Georges Gonthier has formalized a proof of the Four-Color
Theorem\footnote{In 1852.~Francis Guthrie conjectured that every map
  can be colored with at most 4 colors such that no two adjacent
  regions share the same color.} in Coq \cite{gonthier-notices}. The
Four Colour Theorem is famous for being the first long-standing
mathematical problem, analyzed by many famous mathematicians, finally
resolved by a computer program (Appel and Haken
\cite{four-color-appel}). This proof broke new ground because it
involved using IBM 370 assembly language computer programs to carry
out a gigantic case analysis, which could not be performed by hand.
The proof attracted criticism: computer programming is known to be
error-prone, and difficult to relate precisely to the formal statement
of a mathematical theorem. Several attempts to simplify the proofs
were made (e.g., Robertson et al. \cite{four-color-robertson}), number
of cases was reduced and programs were written in C instead of
assembly language. However, all doubts were removed only when Gonthier
employed proof assistant Coq reducing the whole proof to several basic
logical principles.

Another example of a similar kind is the proof of Kepler's
conjecture\footnote{In 1611 Kepler asserted that the so called
  cannonball packing is a densest arrangement of 3-dimensional balls
  of the same size.}. As described by Nipkow et al.
\cite{nipkow-tamegraphs}: ``In 1998. Thomas Hales announced the first
(by now) accepted proof of Kepler's conjecture. It involves 3 distinct
large computations. After 4 years of refereeing by a team of 12
referees, the referees declared that they were 99\% certain of the
correctness of the proof. Dissatisfied with this, Hales started the
informal open-to-all collaborative \emph{flyspeck} project to
formalize the whole proof with a theorem proof.''

\smallskip

In this work, we apply the proof-by-evaluation paradigm to a problem
of verifying FC-families --- a special case of the Frankl's
conjecture. Frankl's conjecture, an elementary and fundamental
statement formulated by P\' eter Frankl in 1979., states that for
every family of sets closed under unions, there is an element
contained in at least half of the sets (or, dually, in every family of
sets closed under intersections, there is an element contained in at
most half of the sets). Up to the best of our knowledge, the problem
is still open. The conjecture has been proved for many special
cases. In particular, it is known to be true for: (i) families of at
most 36 sets\footnote{Unpublished report by Roberts from 1992 claimis
  a similar result for families of at most 40 sets.}
\cite{frankl-lo-faro}; (ii) families of sets such that their union has
at most 11 elements \cite{frankl-bosnjak-markovic}.

FC-families are families for which it is proved that all union closed
families containing them satisfy the Frankl's condition (if the
Frankl's conjecture would be proved, then every family would be an
FC-family).  For example, it can easily be shown that if a family
contains a one-element set, then it satisfies the Frankl's
condition. Similar results holds for any two-element set,
etc. FC-families are important building block for attempting to prove
the Frankl's conjecture since they justify pruning large portions of
the search space.

\paragraph{Related work.}
The Frankl's conjecture has also been formulated and studied as a
question in lattice theory
\cite{frankl-lattices-reinhold,frankl-lattices-abe}.

FC-families have been introduced by Poonen \cite{frankl-poonen} and
further studied by Gao and Yu \cite{frankl-gao-yu}, Vaughan
\cite{frankl-vaughan-1,frankl-vaughan-2,frankl-vaughan-3}, Morris
\cite{frankl-morris}, Markovi\'c \cite{frankl-markovic}, Bo\v snjak
and Markovi\'c \cite{frankl-bosnjak-markovic}, and \v Zivkovi\'c and
Vu\v ckovi\'c \cite{frankl-zivkovic-vuckovic}.

The basic technique used (the Frankl's condition characterization
based on weight functions and shares) is introduced by Poonen
\cite{frankl-poonen} and later successfully used by Bo\v snjak and
Markovi\'c \cite{frankl-markovic,frankl-bosnjak-markovic}, and \v
Zivkovi\' c and Vu\v ckovi\'c \cite{frankl-zivkovic-vuckovic}.

First attempts in using computer-assisted computational approach on
solving special cases of the Frankl's conjecture are described by \v
Zivkovi\' c and Vu\v ckovi\' c
\cite{frankl-zivkovic-vuckovic}. Computations are performed by
(unverified) Java programs. However, in order to increase the level of
trust, Java programs generate certificates that can be checked by
independent tools.

\smallskip

The present paper represent a formalized reformulation of the results
of \v Zivkovi\'c and Vu\v ckovi\'c
\cite{frankl-zivkovic-vuckovic}. All mathematical content is
rigorously formalized within Isabelle/HOL and proofs are mechanically
checked. JAVA programs are reimplemented in a functional language of
Isabelle/HOL and their correctness is formally verified. A clear
separation of mathematical and computational content is done and parts
of the proofs that rely on computations are clearly isolated. Since
the whole formalization is performed and verified within a proof
assistant, there is no need for explicit certificates for statements
proved by computation.

Our main contribution are rigorous, machine-verifiable
proofs\footnote{Corresponding Isabelle/HOL proof documents are
  available from \url{http://argo.matf.bg.ac.rs}} that all FC-families
previously described in the literature are indeed FC-families. Unlike
most pen-and-paper proofs, our proofs follow a uniform approach,
supported by an underlying combinatorial search procedure. The second
contribution is a new type of FC-families: four three-element sets all
contained in a seven-element set.

\paragraph{Background logic and notation.}
Logic and the notation given in this paper will follow Isabelle/HOL.
Isabelle/HOL \cite{isabelle} is a development of Higher Order Logic
(HOL), and it conforms largely to everyday mathematical notation. The
basic types include truth values ($\mathit{bool}$), natural numbers
($\mathit{nat}$) and integers ($\mathit{int}$). Functions can be
defined by recursion (either primitive or general). Sets over type
$\alpha$, type $\alpha\,\mathit{set}$, follow the usual mathematical
conventions\footnote{In a strict type setting, sets containing
  elements of mixed types are not allowed.}. Sets of sets (i.e.,
object of the type $\alpha\,\mathit{set}\,\mathit{set}$) are called
families. Set of all subset for a set $A$ is denoted by $\pow{A}$, and
its number of elements is denoted by $\card{A}$. Lists over type
$\alpha$, type $\alpha\,\mathit{list}$, come with the empty list
$[\,]$, the infix prepend constructor $\#$, the infix $@$ that appends
two lists, and the conversion function $\mathit{set}$ from lists to
sets. N-th element of a list $l$ is denoted by $\nth{l}{n}$. List $[0,
1, \ldots, n-1]$ is denoted by $[0..<n]$. The function ${\sf sort}$
sorts a list, ${\sf listsum}$ calculates its sum, and ${\sf remdups}$
removes duplicate elements. List with no repeated elements are called
distinct. Standard higher order functions ${\sf map}$, ${\sf filter}$,
${\sf foldl}$ are also supported (for details see \cite{isabelle}).

\medskip
All definitions and statements given in this paper are formalized
within Isabelle/HOL. However, in order to make the text accessible to
a more general audience not familiar with Isabelle/HOL, many minor
details are omitted and some imprecisions are introduced (for example,
we used standard symbolics used in related work, although it is clear
that some symbols are ambigous). Statements are grouped into
propositions, lemmas, and theorems. Propositions usually express
simple, technical results and are printed here without proofs.  All
sets and families are considered to be finite and this assumptions
(present in Isabelle/HOL formalization) will not be explicitly stated
in the rest of the paper.

\paragraph{Outline.} The rest of the paper is organized as follows.
In Section \ref{sec:frankl} we give mathematical background on
union-closed families, the Frankl's conjecture and prove main
theoretical results. In Section \ref{sec:search} we formulate the
combinatorial search algorithm, prove its correctness and give its
efficient implementation. In Section \ref{sec:uniform} we introduce
uniform families and techniques used for avoiding symmetries when
analyzing them. In Section \ref{sec:fc} we verify several kinds of
uniform FC-families. Finally, in Section \ref{sec:conclusions} we draw
conclusions and give directions for further work.

\section{Frankl's Families}
\label{sec:frankl}

\subsection{Union Closed Families}
First we give basic definitions of union-closed families, closure
under unions, and operations used to incrementally obtain closed
families.

\begin{definition}
  Let $F$ and $F_c$ be families.

  $F$ is \emph{union closed}, denoted by $\uc{F}$, iff $\forall A \in
  F.\ \forall B \in F.\ A \union B \in F.$
  $F$ is \emph{union closed for $F_c$}, denoted by $\uca{F_c}{F}$, iff
  $\uc{F} \wedge (\forall A \in F.\ \forall B \in F_c.\ A \union B \in F).$

  \emph{Closure of $F$}, denoted by $\closure{F}$, is the minimal
  family of sets (in sense of inclusion) that contains $F$ and is
  union closed.
  \emph{Closure of $F$ for $F_c$}, denoted by $\closurea{F_c}{F}$, is the
  minimal family of sets (in sense of inclusion) that contains $F$ and
  is union closed for $F_c$.

  \emph{Insert and close operation} of set $A$ to family $F$, denoted
  by $\ic{A}{F}$, is the family $F \union\ \{A\}\ \union\ \{A \cup
  B.\ B \in F\}.$
  \emph{Insert and close operation for $F_c$} of set $A$ to family $F$,
  denoted by $\ica{F_c}{A}{F}$, is the family $F \union\ \{A\}\ \union\
  \{A \cup B.\ B \in F\}\ \union\ \{A \cup B.\ B \in F_c\}.$
\end{definition}

\begin{proposition}
\label{prop:closure}\hfill
\vspace{-2mm}
\begin{enumerate}
\item $\closure{F} = \{\Union F'.\ F' \in \pow{F} - \{\emptyset\}\}$
\item $\closure{F \union \{A\}} = \ic{A}{\closure{F}}$, \quad $\closurea{I}{F \union \{A\}} = \ica{I}{A}{\closure{F}}$
\item If $F \subseteq \pow{\Union A}$ and $\uca{A}{F}$ then
  $\uca{\closure{A}}{F}$.
\end{enumerate}
\end{proposition}

\subsection{The Frankl's Condition}
The next definition formalizes the Frankl's condition and the notion
of FC-family.

\begin{definition}
  Family of sets $F$ satisfies the \emph{Frankl's condition} and we
  say that it is a \emph{Frankl's family}, denoted by $\frankl{F}$, if
  it contains an element that occurs in at least half sets in the
  family, i.e., $\frankl{F}\ \equiv\ \exists a.\ a \in \Union F\ \And\
  2 \cdot \cnt{a}{F} \ge \card{F}$, where $\cnt{a}{F}$ denotes
  $\card{\{A \in F.\ a \in A\}}$
  
  Family of sets $F_c$ is \emph{FC-family} if it is proved that every
  union closed family such that $F \supseteq F_c$ is Frankl's.
\end{definition}

\subsection{Family Isomorphisms}
The domain of the family does not play any important role for many
properties related to the Frankl's condition --- many properties are
invariant for domain changes using injective functions (that establish
a kind of isomorphisms between two families). Therefore, in many cases
it suffices to consider only families over canonical domains ---
initial ranges $\setn{n}$ of natural numbers.

\begin{proposition}
  \label{prop:inj}
  Let $F$ be a family of sets and $f$ a function injective on
  $\Union{F}$. Let $F'$ be the image of $F$ under $f$ (then $f$ is a
  bijection between $\Union{F}$ and $\Union{F'}$). \hfill
  \begin{enumerate}
  \item If $a\in \Union{F}$, then $\cnt{a}{F} = \cnt{f(a)}{F'}$.
  \item $\card{F} = \card{F'}$
  \item If $A \in F$ and $A' \in F'$ is the image of $A$ under $f$,
    then $\card{A} = \card{A'}$.
  \item $F$ is union closed if and only if $F'$ is.
  \item $F$ is Frankl's if and only if  $F'$ is.
  \item If $F'$ is an FC-family, then so is $F$.
  \end{enumerate}
\end{proposition}

\subsection{FC Characterization by Weight Functions and Shares}
We describe the central technique for proving that a family is
FC-family, relying on characterizations of the Frankl's condition
using weights and shares.

\begin{definition}
  A function $w: X \rightarrow \mathbb{N}$ is a \emph{weight function
    on} $A \subseteq X$, denoted by $\wf{w}{A}$, iff $\exists a \in
  A.\ w(a) > 0$.  \emph{Weight of a set $A$ wrt.~weight function $w$},
  denoted by $\sw{w}{A}$, is the value $\sum_{a \in A}w(a)$.
  \emph{Weight of a family $F$ wrt.~weight function $w$}, denoted by
  $\fw{w}{F}$, is the value $\sum_{A\in F}\sw{w}{A}$.
\end{definition}

\begin{lemma}
\label{lemma:Frankl_weight}
$\frankl{F} \iff \exists w.\ \wf{w}{(\Union F)}\ \And\ 2 \cdot \fw{w}{F} \;\ge\; \sw{w}{\Union{F}}\cdot\card{F}$
\end{lemma}
\begin{proof}
  Assume $\frankl{F}$ and let $a$ be the element satisfying the
  Frankl's condition. Let $w$ be the weight function assigning 1 to
  $a$ and 0 to all other elements. Since $\fw{w}{F} = \cnt{a}{F}$ and
  $\sw{w}{\Union{F}} = 1$, the statements holds. 
  
  Conversely, suppose that $\neg \frankl{F}$. Then, for every $a \in
  \Union{F}$, $2 \cdot \cnt{a}{F} < \card{F}$. Hence, $2 \cdot \fw{w}{F} =
  \sum_{a \in \Union{F}}w(a)\cdot 2 \cdot \cnt{a}{F}$ $<$ $\card{F}
  \cdot \sum_{a \in \Union{F}}w(a)$ $=$ $\card{F} \cdot
  \sw{w}{\Union{F}}$.
\end{proof}

A concept that will enable a slightly more operative formulation of
the previous characterization is the concept of
\emph{share}\footnote{Note that in order to accommodate for computer
  implementation only integer weights are allowed, and to avoid
  rational numbers share of a set $A$ is defined as $2 \cdot \sw{w}{A}
  - \sw{w}{X}$, instead of $\sw{w}{A} - \sw{w}{X} / 2$ that is used in
  the literature.}.
\begin{definition}
  Let $w$ be a weight function.  \emph{Share of a set $A$ wrt.~$w$ and
    a set $X$}, denoted by $\ss{A}{w}{X}$, is the value $2 \cdot
  \sw{w}{A} - \sw{w}{X}$.  \emph{Share of a family $F$ wrt.~$w$ and a
    set $X$}, denoted by $\fs{F}{w}{X}$, is the value $\sum_{A \in
    F}\ss{A}{w}{X}$.
\end{definition}

\begin{example}
\label{ex:share}
Let $w$ be a function such that $w(a_0) = 1, w(a_1) = 2$, and $w(a) =
0$ for all other elements. $w$ is clearly a weight function. Then,
$\sw{w}{\{a_0, a_1, a_2\}} = 3$ and $\fw{w}{\{\{a_0, a_1\}, \{a_1,
  a_2\}, \{a_1\}\}} = 7$. Also, $\ss{\{a_1, a_2\}}{w}{\{a_0, a_1,
  a_2\}} = 2\cdot \sw{w}{\{a_1, a_2\}} - \sw{w}{\{a_0, a_1, a_2\}} = 4
- 3 = 1,$ and $\fs{\{\{a_0, a_1\}, \{a_1, a_2\},
  \{a_1\}\}}{w}{\{a_0, a_1, a_2\}} = (2\cdot 3 - 3) + (2\cdot 2 - 3) +
(2 \cdot 2 - 3) = 5.$
\end{example}

\begin{proposition}
\label{prop:Family_share}
$\fs{F}{w}{X} = 2 \cdot \fw{w}{F} - \sw{w}{X}\cdot\card{F}$
\end{proposition}

\begin{lemma}
\label{thm:Frankl_Family_share_ge_0}
$\frankl{F} \iff \exists w.\ \wf{w}{(\Union{F})}\ \And\ \fs{F}{w}{(\Union{F})} \ge 0$
\end{lemma}
\begin{proof}
  Follows directly from Proposition \ref{prop:Family_share} and Lemma
  \ref{lemma:Frankl_weight}.
\end{proof}

\paragraph{Hypercubes.}
Sets of a family can be grouped into so called hypercubes.
\begin{definition}
  An $S$-\emph{hypercube} with a base $K$, denoted by $\hc{K}{S}$, is
  the family $\{A.\ K \subseteq A \And A \subseteq K \union
  S\}$. Alternatively, a hypercube can be characterized by $\hc{K}{S}
  = \{K \union A.\ A \in \pow{S}\}$.
\end{definition}

\begin{example}
  \label{ex:hypercube}
  Let $S \equiv \{s_0, s_1\}$, and $K \equiv \{k_0, k_1\}$. If $K'
  \subseteq K$, then all $S$-hypercubes with a base $K'$ are:
  \begin{small}
  \begin{eqnarray*}
    \hc{\{\}}{S} &=& \{\{\}, \{s_0\}, \{s_1\}, \{s_0, s_1\}\}\\
    \hc{\{k_0\}}{S} &=& \{\{k_0\}, \{k_0, s_0\}, \{k_0, s_1\}, \{k_0, s_0, s_1\}\}\\
    \hc{\{k_1\}}{S} &=& \{\{k_1\}, \{k_1, s_0\}, \{k_1, s_1\}, \{k_1, s_0, s_1\}\}\\
    \hc{\{k_0, k_1\}}{S} &=& \{\{k_0, k_1\}, \{k_0, k_1, s_0\}, \{k_0, k_1, s_1\}, \{k_0, k_1, s_0, s_1\}\}\\
  \end{eqnarray*}
  \end{small}
\end{example}
\vspace{-5mm}

Previous example indicates that (disjoint) $S$-hypercubes can span the
whole $\pow{(K \union S)}$. Indeed, this is generally the case.
\begin{proposition} 
(i) $\pow{(K \union S)} = \bigcup_{K' \subseteq K} \hc{K'}{S}$.
(ii) If $K_1$ and $K_2$ are different and disjoint with $S$, then
    $\hc{K_1}{S}$ and $\hc{K_2}{S}$ are disjoint.
\end{proposition}

Families of sets can be separated into (disjoint) parts belonging to
different hypercubes (formed as $\hc{K}{S} \inter F$).
\begin{definition}
  A \emph{hyper-share of a family $F$ wrt.~weight function $w$, the
    hypercube $\hc{K}{S}$ and the set $X$}, denoted by
  $\hs{K}{S}{F}{w}{X}$, is the value $\sum_{A \in \hc{K}{S} \inter
    F}\ss{A}{w}{X}$.
\end{definition}

\begin{example}
  \label{ex:hypershare}
  Let $S$ and $K$ be as in the Example \ref{ex:hypercube}, let $X
  \equiv K \union S$, let $F \equiv \{\{s_0\}, \{s_1\}, \{k_0, s_0\},
  \{k_0, k_1, s_0, s_1\}\}$, and $w(a) = 1$ for all $a \in X$.  Then,
  $\hs{\{\}}{S}{F}{w}{X} = \ss{\{s_0\}}{w}{X} + \ss{\{s_1\}}{w}{X} =
  -4$, $\hs{\{k_0\}}{S}{F}{w}{X} = \ss{\{k_0, s_0\}}{w}{X} = 0$,
  $\hs{\{k_1\}}{S}{F}{w}{X} = 0$, and $\hs{\{k_0, k_1\}}{S}{F}{w}{X}$
  $=$ $\ss{\{k_0, k_1, s_0, s_1\}}{w}{X} = 4$.
\end{example}

Share of a family can be expressed in terms of sum of hyper-shares.
\begin{proposition}
  \label{prop:Family_share_Hyper_sher}
  If $K \union S = \Union{F}$ and $K \inter S = \emptyset$,
  then $\fs{F}{w}{(\Union{F})} = \sum_{K' \subseteq K}
  \hs{K'}{S}{F}{w}{(\Union{F})}$.
\end{proposition}

\begin{lemma}
\label{lemma:Frankl_all_Hyper_share_ge_0}
Let $w$ be a weight function on $\Union{F}$.  If $K \union S =
\Union{F}$, $K \inter S = \emptyset$, and $\forall K' \subseteq
K.\ \hs{K'}{S}{F}{w}{(\Union{F})} \ge 0$, then $\frankl{F}$.
\end{lemma}
\begin{proof}
  Immediate consequence of Proposition \ref{prop:Family_share_Hyper_sher}
  and Lemma \ref{thm:Frankl_Family_share_ge_0}.
\end{proof}

\begin{definition}
  \emph{Projection of a family $F$ onto a hypercube $\hc{K}{S}$},
  denoted by $\hcprj{K}{S}{F}$, is the set $\{A - K.\ A \in \hc{K}{S}
  \inter F\}$.
\end{definition}

\begin{example}
  \label{ex:hcprj}
  Let $K$, $S$ and $F$ be as in Example \ref{ex:hypershare}. Then
  $\hcprj{\{\}}{S}{F} = \{\{s_0\}, \{s_1\}\}$, $\hcprj{\{k_0\}}{S}{F}
  = \{\{s_0\}\}$, $\hcprj{\{k_1\}}{S}{F} = \{\}$, and $\hcprj{\{k_0,
    k_1\}}{S}{F} = \{\{s_0, s_1\}\}$.
\end{example}

\begin{proposition}\label{prop:hcprj}
\hfill
\begin{enumerate}
\item If $K \inter S = \emptyset$ and $K' \subseteq K$, then
  $\hcprj{K'}{S}{F} \subseteq \pow{S}$
\item If $\uc{F}$, then $\uc{(\hcprj{K}{S}{F})}$.
\item If $\uc{F}$, $F_c \subseteq F$, $S = \Union F_c$, $K \inter S =
  \emptyset$, then $\uca{F_c}{(\hcprj{K}{S}{F})}$.
\item If $\forall x \in K.\ w(x) = 0$, then $\hs{K}{S}{F}{w}{X} =
  \fs{\hcprj{K}{S}{F}}{w}{X}$.
\end{enumerate}
\end{proposition}

\paragraph{Union closed extensions.}
The next definition introduces an important notion for checking
FC-families.
\begin{definition}
  \emph{Union closed extensions} of a family $F_c$ are families that
  are created from elements of $F_c$ and are union closed for $F_c$.
  Family of all union closed extensions is denoted by $\uce{F_c}$, and
  $\uce{F_c} \equiv \{F'.\ F' \subseteq \pow{\Union{F_c}} \And
  \uca{F_c}{F'}\}$.
\end{definition}

\begin{lemma}
  \label{lemma:Frankl_Min_Family_share_ge_0}
  Let $F$ be a non-empty union closed family, and let $F_c$ be a
  subfamily (i.e., $F_c \subseteq F$). Let $S$ denote $\Union{F_c}$,
  and let $K$ denote $\Union{F} - \Union{F_c}$.  Let $w$ be a weight
  function on $\Union{F}$, that is zero for all elements of $K$.  If
  shares of all union closed extension of $F_c$ are nonnegative, then
  $F$ is Frankl's, i.e., if $\forall F' \in
  \uce{F_c}.\ \fs{F'}{w}{(\Union{F_c})} \ge 0$, then $\frankl{F}$.
\end{lemma}
\begin{proof}
   Since, $K \union S = \Union F$ and $K \inter S = \emptyset$, by
   Lemma \ref{lemma:Frankl_all_Hyper_share_ge_0}, it suffices to show
   that $\forall K' \subseteq K.\ \hs{K'}{S}{F}{w}{(\Union{F})} \ge
   0$.  Fix $K'$ and assume that $K' \subseteq K$. Since $w$ is zero
   on $K$, by Proposition \ref{prop:hcprj}, it holds that
   $\hs{K'}{S}{F}{w}{(\Union{F})} =
   \fs{\hcprj{K'}{S}{F}}{w}{(\Union{F})}$. On the other hand, since
   $\uc{F}$, $F_c \subseteq F$, and $K \inter S = \emptyset$, by
   Proposition \ref{prop:hcprj} it holds that
   $\uca{F_c}{(\hcprj{K'}{S}{F})}$.  Moreover, $\hcprj{K'}{S}{F}
   \subseteq \pow{S}$, so $\hcprj{K'}{S}{F} \in \uce{F_c}$. Then,
   $\fs{\hcprj{K'}{S}{F}}{w}{(\Union{F_c})} \ge 0$ holds from the
   assumption. However, since $w$ is zero on $K$, it holds that
   $\sw{w}{\Union{F_c}} = \sw{w}{\Union{F}}$ and
   $\fs{\hcprj{K'}{S}{F}}{w}{(\Union{F})} =
   \fs{\hcprj{K'}{S}{F}}{w}{(\Union{F_c})} \ge 0$
\end{proof}

\begin{theorem}
\label{thm:FC_uce_shares_nonneg}
A family $F_c$ is an FC-family if there is a weight function $w$ such
that shares (wrt.~$w$ and $\Union F_c$) of all union closed extension
of $F_c$ are nonnegative.
\end{theorem}
\begin{proof}
  Consider a union-closed family $F \supseteq F_c$. Let $w$ be the
  weight function such that $\forall F' \in
  \uce{F_c}.\ \fs{F'}{w}{(\Union{F_c})} \ge 0$. Let $w'$ be a function
  equal to $w$ on $\Union F_c$ and 0 on other elements. Since $\forall
  F' \in \uce{F_c}.\ \fs{F'}{w'}{(\Union{F_c})} =
  \fs{F'}{w}{(\Union{F_c})}$, Lemma
  \ref{lemma:Frankl_Min_Family_share_ge_0} applies to $F$ and $F$ is
  Frankl's.
\end{proof}

\section{Combinatorial search}
\label{sec:search}

Theorem \ref{thm:FC_uce_shares_nonneg} inspires a procedure for
verifying FC families. It should take a weight function on
$\Union{F_c}$ and check that all union closed extensions of $F_c$ have
nonnegative shares. We will now define a procedure
\emph{SomeShareNegative}, denoted by $\ssn{F_c}{w}$, such that if
$\ssn{F_c}{w} = \bot$, then for all $F' \in \uce{F_c}$ it holds that
$\fs{F'}{w}{(\Union{F_c})} \ge 0$. The heart of this procedure will be
a recursive function $\ssnaux{L}{F_t}{F_c}{w}{X}$ that preforms a
systematic traversal of all union closed extensions of $F_c$, but with
pruning that speeds up the search. If a union closed extension of
$F_c$ has a negative share, it must contain one or more sets with a
negative share. Therefore, a list $L$ of all different subsets of
$\Union{F_c}$ with negative shares is formed and each candidate family
is determined by elements of $L$ that it includes.  A recursive
procedure creates all candidate families by processing elements of $L$
sequentially, either skipping them (in one recursive branch) or
including them into the current candidate family $F_t$ (in the other
recursive branch), maintaining the invariant that the current
candidate family $F_t$ is always union closed. If the current element
of $L$ has been already included in $F_t$ (by earlier closure
operations required to maintain the invariant) the search can be
pruned. If the sum of (negative) shares of the remaining elements of
$L$ is less then the (nonnegative) share of the current $F_t$, then
$F_t$ cannot be extended to a family with a negative share (even in
the extreme case when all the remaining elements of $L$ are included)
so, again, the search can be pruned.

\begin{definition}The function $\ssnaux{L}{F_t}{F_c}{w}{X}$ is defined by 
a primitive recursion (over the structure of the list $L$):
  \begin{eqnarray*}
   \ssnaux{\emptylist}{F_t}{F_c}{w}{X} &\equiv& \fs{F_t}{w}{X} < 0 \\
   \ssnaux{(\cons{h}{t})}{F_t}{F_c}{w}{X} &\equiv& \mathrm{ if\ } \fs{F_t}{w}{X} + \sum_{A \in \cons{h}{t}}\ss{A}{w}{X} \ge 0\mathrm{\ then\ } \bot \\
    && \mathrm{else\ if\ }\ssnaux{t}{F_t}{F_c}{w}{X} \mathrm{\ then\ } \top\\
    && \mathrm{else\ if\ }h \in F_t \mathrm{\ then\ } \bot\\
    && \mathrm{else\ } \ssnaux{t}{(\ica{F_c}{h}{F_t})}{F_c}{w}{X}
  \end{eqnarray*}

  Let $L$ be a distinct list such that its set is $\{A.\ A \in
  \pow{\Union{F_c}} \And \ss{A}{w}{X} < 0\}$.
  $$\ssn{F_c}{w} \equiv \ssnaux{L}{\emptyset}{\closure{F_c}}{w}{(\Union{F_c})}$$
\end{definition}

Next we prove the soundnes of the $\ssn{F_c}{w}$ function.

\begin{lemma}
  \label{lemma:ssnaux_correct}
  If (i) $\ssnaux{L}{F_t}{F_c}{w}{X} = \bot$, (ii) for all elements
  $A$ in $L$ it holds that $\ss{A}{w}{X} < 0$, (iii) for all $A \in F'
  - F_t$, if $\ss{A}{w}{X} < 0$, then $A$ is in $L$, (iv) $F'
  \supseteq F_t$, and (v) $\uca{F_c}{F'}$, then $\fs{F'}{w}{X} \ge 0$.
\end{lemma}
\begin{proof}
  The proof is by induction.
  First, note that 
  \begin{small}
  \begin{equation}
    \label{eq:1}
    \fs{F'}{w}{X} = \sum_{A \in F'}\ss{A}{w}{X} =
    \sum_{A \in F_t}\ss{A}{w}{X} + \sum_{A \in F' - F_t}\ss{A}{w}{X}.
  \end{equation}
  \end{small}

  Consider the base case of $L = \emptylist$. Since
  $\ssnaux{\emptylist}{F_t}{F_c}{w}{X} = \bot$, it holds that $\sum_{A
    \in F_t}\ss{A}{w}{X} = \fs{F_t}{w}{X} \ge 0$ and first term in
  (\ref{eq:1}) is nonnegative. If there were some $A \in F'-F_t$ such
  that $\ss{A}{w}{X} < 0$, then, from the assumptions it would be in
  $L$, which is impossible since $L$ is empty. Therefore, the second
  term in (\ref{eq:1}) is also nonnegative which completes the proof.

  \smallskip
  Consider the inductive step, and assume that $L \equiv \cons{h}{t}$.
  
  \smallskip
  First consider the case when $\fs{F_t}{w}{X} + \sum_{A \in
    \cons{h}{t}}\ss{A}{w}{X} \ge 0$. Let $P$ denote the set $\{A.\ A
  \in F' - F_t \And \ss{A}{w}{X} \ge 0\}$, and let $N$ denote the set
  $\{A.\ A \in F' - F_t \And \ss{A}{w}{X} < 0\}$. 
  Since, by assumptions, all elements of $N$ are in $L \equiv
  \cons{h}{t}$, and since, by assumptions, all shares of $\cons{h}{t}
  - N$ are negative, it holds that
  \begin{small}
  \begin{equation}
    \label{eq:2}
    \sum_{A \in \cons{h}{t}}\ss{A}{w}{X} = \sum_{A \in N}\ss{A}{w}{X} +
  \sum_{A \in \cons{h}{t} - N}\ss{A}{w}{X} \le \sum_{A \in
    N}\ss{A}{w}{X}.
  \end{equation}
  \end{small}

  It holds that $\sum_{A \in F' - F_t}\ss{A}{w}{X} = \sum_{A \in
    P}\ss{A}{w}{X} + \sum_{A \in N}\ss{A}{w}{X}.$
  Therefore, since all shares of $P$ are nonnegative, from
  (\ref{eq:1}) and (\ref{eq:2}) and the assumption of the current case
  it holds that

  \begin{small}
  $$\fs{F'}{w}{X} \ge \sum_{A \in F_t}\ss{A}{w}{X} + \sum_{A \in N}\ss{A}{w}{X} \ge \fs{F_t}{w}{X} + \sum_{A \in \cons{h}{t}}\ss{A}{w}{X} \ge 0.$$
  \end{small}

  \smallskip Next, consider the case when $\fs{F_t}{w}{X} + \sum_{A
    \in \cons{h}{t}}\ss{A}{w}{X} < 0$. Since, by assumptions,
  $\ssnaux{(\cons{h}{t})}{F_t}{F_c}{w}{X} = \bot$, by the definition of
  ${\sf ssn}$ it must hold that $\ssnaux{t}{F_t}{F_c}{w}{X} = \bot$.

  \smallskip Consider the case when $h \in F_t$ or $h \notin F'$. Then
  $h \notin F' - F_t$. The conclusion follows by induction hypothesis
  for the recursive call $\ssnaux{t}{F_t}{F_c}{w}{X}$, since all
  assumptions are satisfied. Indeed, all elements of $F' - F_t$ with
  negative shares must be in $t$, since $h \notin F' - F_t$, and other
  assumptions are trivially satisfied.
  
  \smallskip Finally, consider the case when $h \notin F_t$ and $h \in
  F'$. The conclusion follows by induction hypothesis for the
  recursive call $\ssnaux{t}{(\ica{F_c}{h}{F_t})}{F_c}{w}{X}$, since all
  assumptions are satisfied for this call. Indeed, in this case
  $\ssnaux{(\cons{h}{t})}{F_t}{F_c}{w}{X} =
  \ssnaux{t}{(\ica{F_c}{h}{F_t})}{F_c}{w}{X}$ and the left hand side is
  $\bot$ from the current assumptions. All elements of $F' -
  \ica{F_c}{h}{F_t}$ with negative shares must be in $t$. Indeed, this
  holds since $F_t \subseteq \ica{F_c}{h}{F_t}$, and $h \in
  \ica{F_c}{h}{F_t}$, and since all elements of $F' - F_t$ with
  negative shares are in $\cons{h}{t}$. It holds that
  $\ica{F_c}{h}{F_t} \subseteq F'$ since $F_t \subseteq F'$, $h \in
  F'$ and $\uca{F_c}{F'}$. Other assumptions trivially hold.
\end{proof}

\begin{theorem}
  \label{lemma:ssn_correct}
  If $\ssn{F_c}{w} = \bot$ and $F' \in \uce{F_c}$ then
  $\fs{F'}{w}{(\Union{F_c})} \ge 0$.
\end{theorem}
\begin{proof}
  Fix $F'$ from $\uce{F_c}$. Then $F' \subseteq \pow{\Union{F_c}}$ and
  $\uca{F_c}{F'}$. Let $L$ be a distinct list such that its set is
  $\{A.\ A \in \pow{\Union{F_c}} \And \ss{A}{w}{X} < 0\}$. From
  $\ssn{F_c}{w} = \bot$ and the definition of ${\sf ssn}$ it holds
  that $\ssnaux{L}{\emptyset}{\closure{F_c}}{w}{(\Union{F_c})} =
  \bot$.  All assumptions of Lemma \ref{lemma:ssnaux_correct}
  apply. Indeed, for all $A$ in $L$, $\ss{A}{w}{(\Union{F_c})} <
  0$. For all $A$ in $F' - \emptyset$, if $\ss{A}{w}{(\Union{F_c})} <
  0$, then, since $F' \subseteq \pow{\Union{F_c}}$, $A$ is in
  $L$. $\emptyset \subseteq F'$. Since $\uca{F_c}{F'}$, by Proposition
  \ref{prop:closure}, it holds that
  $\uca{\closure{F_c}}{F'}$. Therefore, $\fs{F'}{w}{(\Union{F_c})} \ge
  0$ holds.
\end{proof}

Apart from being sound, the procedure can also be shown to be
complete. Namely, it could be shown that if $\ssn{F_c}{w} = \top$,
then there is an $F' \in \uce{F_c}$ such that
$\fs{F'}{w}{(\Union{F_c})} < 0$. This comes from the invariant that
the current family $F_t$ in the search is always in $\uce{F_c}$, which
is maintained by taking the closure $\ica{F_c}{h}{F_t}$ whenever an
element $h$ is added. Since this aspect of the procedure is not
relevant for the rest of the proofs, it will not be formally stated
nor proved.

\subsection{Efficient implementation}
In order to obtain executability and increase efficiency, a series of
refinements of $\ssn{F}{w}$ is done. Each refined version introduces a
new implementation feature that makes it more efficient than the
previous one, but still equivalent with it.

First, a function cannot operate on families of sets. Without loss of
generality, it suffices only to consider families of sets of natural
numbers. Sets of natural numbers are represented by natural number
codes. A set $A$ is represented by the code $\tilde{A} = \sum_{k \in
  A} 2^k$. Families of sets of natural numbers $F$ are represented by
(distinct) lists of natural number codes $\tilde{F}$. This
representation will be referred to as \emph{list-of-nats}
representation (e.g., $F = \{\{0, 1\}, \{1, 2\}, \{0, 1, 2\}\}$ is
represented by the list-of-nats $\tilde{F} = [3, 6, 7]$).  Basic set
operations have their corresponding list-of-nat counterparts.
\begin{itemize}
\item The union of two sets $\union$ corresponds to bitwise
  disjunction (denoted by $\sqcup$).  It holds that if $C =
  A\;\union\;B$, then $\tilde{C} = \tilde{A}\;\sqcup\;\tilde{B}$.
\item Adding a set $A$ to a family of sets $F$ (i.e., $A\;\union\;F$)
  corresponds to the operation (also denoted by $\sqcup$) that
  prepends $\tilde{A}$ to $\tilde{F}$, but only if it is not already
  present, i.e., by: $\mathrm{if\ } \tilde{A} \in \tilde{F}
  \mathrm{\ then\ } \tilde{F} \mathrm{\ else\ }
  \cons{\tilde{A}}{\tilde{F}}$. It holds that if $F' = A \union F$,
  then $\tilde{F'} = \tilde{A} \sqcup \tilde{F}$.
\item Union of two families (i.e., $F' \union F$), also denoted by
  $\sqcup$, is performed by iteratively adding sets from one family to
  another, i.e., as $\foldl{(\lambda\ \tilde{A}\ \tilde{F}.\
  \tilde{A} \sqcup \tilde{F})}{\tilde{F}}{\tilde{F'}}$. It holds that
  if $F'' = F \union F'$, then $\tilde{F''} = \tilde{F} \sqcup
  \tilde{F'}$.
\item Adding a set $A$ to all members of a family of sets $F$ (i.e.,
  $\{A \union B.\ B \in F\}$), denoted by
  $[\tilde{A}\;\sqcup\;\tilde{B}.\ \tilde{B} \in \tilde{F}]$, is
  performed by $\map{(\lambda\ \tilde{B}.\ \tilde{A} \sqcup
  \tilde{B})}{\tilde{F}}$. It holds that if $F' = \{A \union B.\ B \in
  F\}$, then $\tilde{F'} = [\tilde{A}\;\sqcup\;\tilde{B}.\ \tilde{B}
  \in \tilde{F}]$.
\item Insert and close for $F$ (i.e., $\ica{F_c}{a}{F}$), denoted by
  $\tilde{{\sf ic}}$, is computed as
  $([\tilde{A}]\ @\ [\tilde{A}\;\sqcup\;\tilde{B}.\ \tilde{B} \in
  \tilde{F}]\ @\ [\tilde{A}\;\sqcup\;\tilde{B}.\ \tilde{B} \in
  \tilde{F_c}] )\ \sqcup\ \tilde{F}$. It holds that if $F' =
  \ica{F_c}{a}{F}$, then $\tilde{F'} =
  \icatilde{\tilde{F_c}}{\tilde{a}}{\tilde{F}}.$
\end{itemize}

Important optimization to the basic $\ssn{F_c}{w}$ procedure is to
avoid repeated computations of family shares (both for the elements of
the list $L$ and the current family $F_t$). So, instead of accepting a
list of families of sets $L$, and the current family of sets $F_t$,
the function is modified to accept a list of ordered pairs where first
component is a list-of-nats representation of corresponding element of
$L$, and the second component is its share (wrt.~$w$ and $X$), and to
accept an ordered pair $(\tilde{F_t}, s_t)$ where $\tilde{F_t}$ is the
list-of-nats representation of $F_t$, and $s_t$ is its family share
(wrt.~$w$ and $X$). The summation of shares of elements in $L$ is also
unnecessarily repeated. It can be avoided if the sum ($s_l$) is passed
trough the function.

\begin{small}
  \begin{eqnarray*}
    \ssnaux{(\emptylist, 0)}{(\tilde{F_t}, s_t)}{\tilde{F_c}}{w}{X} &\equiv& s_t < 0 \\
    \ssnaux{(\cons{(\tilde{h}, s_h)}{t},\;s_l)}{(\tilde{F_t},\;s_t)}{\tilde{F_c}}{w}{X} &\equiv& \mathrm{ if\ } s_t + s_l \ge 0\mathrm{\ then\ } \bot \\
    && \mathrm{else\ if\ }\ssnaux{(t,\;s_l - s_h)}{(\tilde{F_t},\;s_t)}{\tilde{F_c}}{w}{X} \mathrm{\ then\ } \top\\
    && \mathrm{else\ if\ }\tilde{h} \in \tilde{F_t} \mathrm{\ then\ } \bot\\
    && \mathrm{else\ let\ } \tilde{F_t}' = \icatilde{\tilde{F_c}}{\tilde{h}}{\tilde{F_t}};\ s_t' = \fs{\tilde{F_t}'}{w}{X} \mathrm{\ in}\\
    && \quad  \ssnaux{(t, ls - s_h)}{(\tilde{F_t}',s_t'\;)}{\tilde{F_c}}{w}{X}
  \end{eqnarray*}
\end{small}

Another source of inefficiency is the calculation of
$\fs{\tilde{F_t}'}{w}{X}$. If performed directly based on the
definition of family share for $\tilde{F_t}'$, the sum would contain
shares of all elements from $\tilde{F_t}$ and of all elements that are
added to $\tilde{F_t}$ when adding $\tilde{h}$ and closing for
$\tilde{F}$. However, it is already known that the sum of shares for
elements of $\tilde{F_t}$ is $s_t$ and the implementation could
benefit from this fact. Also, calculating shares of sets that are
added to $\tilde{F_t}$ can be made faster. Namely, it happens that set
share of a same set is calculated over and over again in different
parts of the search space. So, it is much better to precompute shares
of all sets from $\pow{X}$ and store them in a lookup table that will
be consulted each time a set share is needed. Note that in this case
there is no more need to pass the function $w$ itself, nor the domain
$X$, but only the lookup table, denoted by $s_w$.

\vspace{-1mm}
\begin{small}
  \begin{eqnarray*}
    \ssnauxa{(\emptylist, 0)}{(\tilde{F_t}, s_t)}{\tilde{F}_c}{s_w} &\equiv& s_t < 0 \\
    \ssnauxa{(\cons{(\tilde{h}, s_h)}{t},\;s_l)}{(\tilde{F_t},\;s_t)}{\tilde{F}_c}{s_w} &\equiv& \mathrm{ if\ } s_t + s_l \ge 0\mathrm{\ then\ } \bot \\
    && \mathrm{else\ if\ }\ssnauxa{(t,\;s_l - s_h)}{(\tilde{F_t},\;s_t)}{\tilde{F}_c}{s_w} \mathrm{\ then\ } \top\\
    && \mathrm{else\ if\ }\tilde{h} \in \tilde{F_t} \mathrm{\ then\ } \bot\\
    && \mathrm{else\ } \ssnauxa{(t, s_l - s_h)}{(\icaa{\tilde{F}_c}{\tilde{h}}{(\tilde{F_t}, s_t)}{s_w})}{\tilde{F}_c}{s_w}\\
  \icaa{\tilde{F}_c}{\tilde{h}}{(\tilde{F_t}, s_t)}{s_w} &\equiv& \mathrm{let\ }\ add\ = \ [\tilde{h}]\ @\
       [\tilde{h}\;\sqcup\;\tilde{A}.\ \tilde{A} \in \tilde{F_t}]\ @\
  [\tilde{h}\;\sqcup\;\tilde{A}.\ \tilde{A} \in \tilde{F}_c];\\
& & \qquad add\ =\ \filter{(\lambda \tilde{A}.\ \tilde{A}\notin \tilde{F})}{(\remdups{add})}\ \textrm{in}\\
& & (add\; @\;\tilde{F},\ s + \listsum{(\map{s_w}{add})})\\
\end{eqnarray*}
\end{small}
\vspace{-9mm}

It is shown that this implementation is (in some sense) equivalent to
the starting, abstract one. This proof is technically involved, but
conceptually uninteresting so we omit it in the text.

\section{Uniform $nkm$-families}
\label{sec:uniform}

Most FC-families that are considered in this paper are \emph{uniform},
i.e., consist of sets having the same number of elements.

\begin{definition}
  A family of sets $F$ is a \emph{uniform $nkm$-family} if it
  contains $m$ different sets, each containing $k$ elements and their
  union has at most $n$ elements. Uniform $nkm$-family is
  \emph{natural} if its union is contained in $\setn{n}$.
\end{definition}
  
Within the Isabelle/HOL implementation, natural $nkm$-families will be
represented by \emph{$nkm$-lists} --- (lexicografically) sorted,
distinct lists of length $m$ containing sorted, distinct lists of
length $k$ with all elements contained in $\setn{n}$. To simplify
presentation, we will identify natural $nkm$-families with their
corresponding $nkm$-lists. Assuming that the Isabelle/HOL function
$\combine{l}{k}$ generates all sorted $k$-element sublists of a sorted
list $l$, all $nkm$-lists for given $n$, $k$ and $m$ can be generated
by $\families{n}{k}{m} \equiv \combine{(\combine{[0..<n]}{k})}{m}$.

\paragraph{Symmetries.}

Often one uniform $nkm$-family can be obtained from the other by
permuting its elements (e.g., $\{\{a_0, a_1, a_2\}, \{a_1 , a_3 ,
a_4\}, \{a_2, a_3, a_4\}\}$ can be obtained from $\{\{a_0, a_1 ,
a_2\}, \{a_0, a_1 , a_3\}, \{a_2, a_3, a_4 \}\}$ by the permutation
$(a_0,$ $a_1,$ $a_2,$ $a_3,$ $a_4)$ $\mapsto$ $(a_3, a_4, a_1, a_2,
a_0)$). Applying permutations on sets and families can be implemented
in Isabelle/HOL by the functions $\permuteset{A}{p} \equiv
\sort{(\map{(\lambda x.\ \nth{p}{x})}{A})}$ and $\permutefamily{F}{p}
\equiv \sort{(\map{{\sf perm\_set}}{F})}$. Permutations establish
bijections between natural uniform families:
\begin{proposition}
  \label{prop:permfam}
  If $p$ is a permutation of $\listn{n}$ and $F$ is a natural uniform
  family, then $\permutefamily{F}{p}$ is also natural uniform family
  and there is a bijection between $F$ and $\permutefamily{F}{p}$.
\end{proposition}

Since, by Proposition \ref{prop:inj}, FC-families are preserved under
bijections (isomorphisms), to check if all elements of a given list of
$nkm$-families $\mathcal{F}$ are FC-families, many elements need not
be considered. Indeed, it suffices to consider only a list (denoted by
$\nef{\mathcal{F}}{P}$) of its non-equivalent representatives (under a
given list of permutations $P$). Computation of such representatives
can start from the given list $\mathcal{F}$, choose its arbitrary
member for a representative, remove it and all its permuted variants
from the lists, and repeat this sieving process until the list becomes
empty. Isabelle/HOL implementation of this procedure can be given by:

\begin{small}
\begin{eqnarray*}
  \nefa{\mathcal{F}}{P}{\mathcal{F}_r} &\equiv& \mathrm{case}\ \mathcal{F}\ \mathrm{of}\ \emptylist \Rightarrow \mathcal{F}_r\\
  &|& \cons{F}{\underrule} \Rightarrow \mathrm{let}\ \mathcal{F}_F^P = \remdups{(\map{(\lambda\ p.\ \permutefamily{F}{p})}{P})}\ \mathrm{in}\\
  && \qquad\qquad\nefa{(\filter{(\lambda\ F.\ F \notin \mathcal{F}_F^P)}{\mathcal{F}})}{P}{(\cons{F}{\mathcal{F}_r})}\\
  \nef{\mathcal{F}}{P} &\equiv& \nefa{\mathcal{F}}{P}{\emptylist}
\end{eqnarray*}
\end{small}

The following lemma proves the correctness of this implementation.

\begin{lemma}
  \label{lemma:nef}
  If $P$ is a list of permutations of $\listn{n}$ and if $\mathcal{F}$
  is a list of natural $nkm$-families, then for each element $F \in
  \mathcal{F}$ there is an $F' \in \nef{\mathcal{F}}{P}$ such there is
  a bijection between $F$ and $F'$.
\end{lemma}
\begin{proof}
  First, note that the function $\nefa{\mathcal{F}}{P}{\mathcal{F}_r}$
  is monotone, i.e., $\mathcal{F}_r \subseteq
  \nefa{\mathcal{F}}{P}{\mathcal{F}_r}$.

  By induction, we show that if the assumptions hold for $\mathcal{F}$
  and $P$, then for each element $F \in \mathcal{F}$ there is an
  element $F' \in \nefa{\mathcal{F}}{P}{\mathcal{F}_r}$ such there is
  a bijection between $F$ and $F'$.

  In the base case, when $\mathcal{F}$ is empty, the statement
  trivially holds.

  Assume that $\mathcal{F} \equiv \cons{F}{\mathcal{F}'}$. Let
  $\mathcal{F}_F^P$ denote all different families obtained by
  permuting $F$ by all elements of $P$ (i.e., $\mathcal{F}_F^P \equiv
  \remdups{(\map{(\lambda\ p.\ \permutefamily{F}{p})}{P})}$) and let
  $\mathcal{F}^-$ denote what remains of $\mathcal{F}$ when those are
  removed (i.e., $\mathcal{F}^- \equiv \filter{(\lambda\ F.\ F \notin
    \mathcal{F}_F^P)}{\mathcal{F}}$. It holds that
  $\nefa{\mathcal{F}}{P}{\mathcal{F}_r} =
  \nefa{\mathcal{F}^-}{P}{(\cons{F}{\mathcal{F}_r})}$.

  Let $F'$ be an arbitrary element from $\mathcal{F}$. Since
  $\mathcal{F} = \cons{F}{\mathcal{F}'}$, either $F' = F$ or $F' \in
  \mathcal{F}'$.

  Assume that $F' = F$. By monotonicity it holds that $F \in
  \nefa{\mathcal{F}}{P}{\mathcal{F}_r}$, so $F$ is an element from
  $\nefa{\mathcal{F}}{P}{\mathcal{F}_r}$ such that there is a
  bijection (identity function) between $F'$ and it.

  Assume that $F' \in \mathcal{F}'$.

  Consider the case when $F' \in \mathcal{F}_F^P$. Then there is $p
  \in P$ such that $F' = \permutefamily{F}{p}$. Since $F' \in
  \mathcal{F}$ is natural and $p \in P$ is a permutation of
  $\listn{n}$, by Proposition \ref{prop:permfam}, there is a bijection
  between $F$ and $F'$. Since, by monotonicity, it holds that $F \in
  \nefa{\mathcal{F}}{P}{\mathcal{F}_r}$, $F$ is an element in
  $\nefa{\mathcal{F}}{P}{\mathcal{F}_r}$ such that there is a
  bijection between $F'$ and it.

  Consider the case when $F' \notin \mathcal{F}_F^P$. Then $F' \in
  \mathcal{F}^-$. By inductive hypothesis for the call
  $\nefa{\mathcal{F}^-}{P}{(\cons{F}{\mathcal{F}_r})}$, there is an
  element $F''$ in $\cons{F}{\mathcal{F}_r}$ such that there is a
  bijection between $F'$ and it. By monotonicity, $F'' \in
  \cons{F}{\mathcal{F}_r} \subseteq
  \nefa{\mathcal{F}^-}{P}{(\cons{F}{\mathcal{F}_r})} =
  \nefa{\mathcal{F}}{P}{\mathcal{F}_r}$, so the statement holds.
\end{proof}

Finally, the following lemma shows that only non-equivalent
representatives need to be considered when checking FC-families.

\begin{lemma}
  \label{lemma:FC_family_non_equivalent_families}
  Let $\mathcal{F} \subseteq \families{n}{k}{m}$ and $P \subseteq
  \permute{\listn{n}}$. If all families represented by elements of
  $\nef{\mathcal{F}}{P}$ are FC-families, then all families represented by
  elements of $\families{n}{k}{m}$ are FC-families.
\end{lemma}
\begin{proof}
  Let $F \in \families{n}{k}{m}$. By Lemma \ref{lemma:nef} there is an
  $F' \in \nef{\mathcal{F}}{P}$ and a bijection between $F$ and
  $F'$. So, $F'$ is an FC-family, and by Proposition \ref{prop:inj}, so
  is $F$.
\end{proof}

\section{FC-families verified}
\label{sec:fc}
Having established all the necessary mathematics, in this Section we
prove that certain uniform families are FC-families (mainly by
performing verified calculations). First, we calculate non-equivalent
representatives for $\families{5}{3}{3}$, $\families{6}{3}{4}$, and
$\families{7}{3}{4}$.

\begin{lemma}
\label{lemma:calc_nef}
The first column of Table \ref{table:ssn} contains (respectively)
all elements of:

\begin{small}
  $\nef{\families{5}{3}{3}}{\permute{[0..<5]}}$,

  $\nef{(\filter{(\lambda F.\ \neg {\sf check_{533}\ F})}{\families{6}{3}{4}})}{\permute{[0..<6]}}$,

  $\nef{(\filter{(\lambda F.\ \neg {\sf check_{533}}\ F \wedge \neg {\sf check_{634}}\ F)}{\families{7}{3}{4}})}{\permute{[0..<7]}},$
\end{small}

\noindent where $\permute{l}$ is the function that generates all
permutations of a list $l$, ${\sf check_{533}}$ is a function that
checks if any 3 of the 4 given 3-element sets are have their union
contained in a 5-element set, and ${\sf check_{634}}$ is a function
that checks if the union of 4 given 3-element sets is contained in a
6-element set.\footnote{Formal definition of these functions is not
  given here and is available in the Isabelle/HOL proof documents,
  along with correctness arguments.}
\end{lemma}
\begin{proof}
  By calculations performed by a computer.
\end{proof}

\newcommand{\wt}[2]{#1 \mapsto #2}
\begin{table}[!t]
  \centering
  \begin{footnotesize}
  \begin{tabular}{c|l}
    $F_c$ &\  $w$\\\hline
    $[[0, 1]]$ &\ $\wt{0}{1}, \wt{1}{1}$ \\\hline
    $[[0, 1, 2], [0, 1, 3], [2, 3, 4]]$ &\  $\wt{0}{2}, \wt{1}{2}, \wt{2}{2}, \wt{3}{2}, \wt{4}{1}$\\
    $[[0, 1, 2], [0, 1, 3], [0, 2, 4]]$ &\  $\wt{0}{6}, \wt{1}{5}, \wt{2}{5}, \wt{3}{3}, \wt{4}{3}$\\
    $[[0, 1, 2], [0, 1, 3], [0, 2, 3]]$ &\  $\wt{0}{1}, \wt{1}{1}, \wt{2}{1}, \wt{3}{1}$\\
    $[[0, 1, 2], [0, 1, 3], [0, 1, 4]]$ &\  $\wt{0}{3}, \wt{1}{3}, \wt{2}{2}, \wt{3}{2}, \wt{4}{2}$\\\hline
    $[[0, 1, 2], [0, 3, 4], [1, 3, 5], [2, 4, 5]]$ &\  $\wt{0}{1}, \wt{1}{1}, \wt{2}{1}, \wt{3}{1}, \wt{4}{1}, \wt{5}{1}$\\
    $[[0, 1, 2], [0, 1, 3], [2, 4, 5], [3, 4, 5]]$ &\  $\wt{0}{1}, \wt{1}{1}, \wt{2}{1}, \wt{3}{1}, \wt{4}{1}, \wt{5}{1}$\\\hline
    $[[0, 1, 2], [0, 3, 4], [1, 3, 5], [2, 4, 6]]$ &\  $\wt{0}{2}, \wt{1}{2}, \wt{2}{2}, \wt{3}{2}, \wt{4}{2}, \wt{5}{1}, \wt{6}{1}$\\
    $[[0, 1, 2], [0, 3, 4], [0, 5, 6], [1, 3, 5]]$ &\  $\wt{0}{2}, \wt{1}{1}, \wt{2}{1}, \wt{3}{1}, \wt{4}{1}, \wt{5}{1}, \wt{6}{1}$\\
    $[[0, 1, 2], [0, 1, 3], [2, 4, 5], [4, 5, 6]]$ &\  $\wt{0}{3}, \wt{1}{3}, \wt{2}{4}, \wt{3}{2}, \wt{4}{3}, \wt{5}{3}, \wt{6}{2}$ \\
    $[[0, 1, 2], [0, 1, 3], [2, 4, 5], [3, 4, 6]]$ &\  $\wt{0}{3}, \wt{1}{3}, \wt{2}{3}, \wt{3}{3}, \wt{4}{2}, \wt{5}{1}, \wt{6}{1}$\\
    $[[0, 1, 2], [0, 1, 3], [0, 4, 5], [4, 5, 6]]$ &\  $\wt{0}{6}, \wt{1}{4}, \wt{2}{3}, \wt{3}{3}, \wt{4}{4}, \wt{5}{4}, \wt{6}{2}$\\
    $[[0, 1, 2], [0, 1, 3], [0, 4, 5], [2, 4, 6]]$ &\  $\wt{0}{3}, \wt{1}{2}, \wt{2}{3}, \wt{3}{1}, \wt{4}{3}, \wt{5}{2}, \wt{6}{2}$\\
    $[[0, 1, 2], [0, 1, 3], [0, 4, 5], [1, 4, 6]]$ &\  $\wt{0}{2}, \wt{1}{2}, \wt{2}{1}, \wt{3}{1}, \wt{4}{1}, \wt{5}{1}, \wt{6}{1}$\\
    $[[0, 1, 2], [0, 1, 3], [0, 4, 5], [0, 4, 6]]$ &\  $\wt{0}{2}, \wt{1}{1}, \wt{2}{1}, \wt{3}{1}, \wt{4}{1}, \wt{5}{1}, \wt{6}{1}$
  \end{tabular}
  \end{footnotesize}
  \caption{Families and weights}
  \label{table:ssn}
\end{table}

Next, we show that all these representatives have non-negative shares.
\begin{lemma}
  \label{lemma:calc_ssn}
  For all $F_c$ and $w$ given in Table \ref{table:ssn}, it holds that
  $\ssn{\tilde{F_c}}{w} = \bot$.
\end{lemma}

\begin{proof}
  By calculations performed by a computer.
\end{proof}

Finally, the main result can be easily proved.

\begin{theorem} The following are FC-families:
  \begin{enumerate}
  \item all families containing one 1-element set (i.e., $\{\{a\}\}$);
  \item all families containing one 2-element set (i.e., $\{\{a,
    b\}\}$, for $a \neq b$);
  \item all families containing 3 3-element sets whose union
    is contained in a 5-element set (i.e., uniform $533$-families);
  \item all families containing 4 3-element sets whose union
    is contained in a 6-element set (i.e., uniform $634$-families);
  \item all families containing 4 3-element sets whose union
    is contained in a 7-element set (i.e., uniform $734$-families).
  \end{enumerate}
\end{theorem}
\begin{proof}
The case 1~trivially holds (since for each family member $A$ that does
not contain $a$, there is a member $A \union \{a\}$ that contains
$a$).

Other proofs are based on the techniques described in this paper. By
Proposition \ref{prop:inj} it suffices to consider only families $F$
such that $\Union F \subseteq {\setn{n}}$. All families corresponding
to rows in Table \ref{table:ssn} are FC-families. Indeed, for each
$F_c$ and $w$ given in a table row, by Lemma \ref{lemma:calc_ssn} it
holds that $\ssn{F_c}{w}$. Therefore, by Lemma \ref{lemma:ssn_correct}
for all $F' \in \uce{F_c}$ it holds that $\fs{F'}{w}{(\Union{F_c})}
\ge 0$. Then, $F_c$ is FC-family by Theorem
\ref{thm:FC_uce_shares_nonneg}.

In the case 2~this completes the proof. 

In the case 3~the statement holds by Lemma
\ref{lemma:FC_family_non_equivalent_families}, since, by Lemma
\ref{lemma:calc_nef} four rows given in Table \ref{table:ssn}
correspond to four non-equivalent families.

To show the case 4, let $F_c$ be any family containing 4 3-element
sets whose union is contained in $\{0, 1, \ldots, 5\}$ and let $F$ be
a union-closed family such that $F \supseteq F_c$. If ${\sf
  check}_{533}\ F_c$ holds (i.e., if union of any 3 members of $F_c$
is contained in a 5-element set), then $F$ is Frankl's by case 3. If
$\neg {\sf check}_{533}\ F_c$ holds, then $F_c$ is in
$\filter{(\lambda F. \neg {\sf
    check}_{533}\ F)}{\families{6}{3}{4}}$. The statement then holds
by Lemma \ref{lemma:FC_family_non_equivalent_families}, since, by
Lemma \ref{lemma:calc_nef} two rows given in Table \ref{table:ssn}
correspond to two non-equivalent families of $\filter{(\lambda F. \neg
  {\sf check}_{533}\ F)}{\families{6}{3}{4}}$.

The case 5 is proved similarly, using the proofs for both the case 3
and the case 4.
\end{proof}

\section{Conclusions and further work}
\label{sec:conclusions}
In this paper, we have formalized (within Isabelle/HOL) a
computer-assisted approach of \v Zivkovi\' c and Vu\v ckovi\' c for
verifying FC-families. Well-known FC-families are confirmed and a
new uniform FC-family is discovered.

The Isabelle/HOL formalization has around 260KB of data organized into
around 6500 lines of Isabelle/Isar proof text. Ratio between the size
of the formalization and the size of the corresponding pen and paper
proof (DeBruijn index) is estimated at around 5.5. Total time required
to do the formalization is very roughly estimated at around 200
man/hours (25 full working days spread over a period of around 8
months).

Total proof checking time of Isabelle/HOL takes around 28 minutes on a
notebook PC with 2.1GHz Intel/Pentium CPU and 4GB RAM. The major
fraction of this time (around 23 minutes) is spent in the
combinatorial search. Checking Lemma \ref{lemma:calc_ssn} consumes
most of this time, and its last 8 cases (related to the uniform-734
families) alone take 22.8 minutes. This is quite long compared to the
original JAVA programs (that perform the whole combinatorial search in
around 1 minute), but still bearable. The big difference is due to the
use of machine-integers supporting atomic bitwise-or in JAVA and the
use of big-integers that do not support atomic bitwise-or in
Isabelle/ML. The search time could be reduced if machine-integers were
also used in Isabelle/ML. In a simple approach, the code generator
could be instructed to replace mathematical integers in the
formalization by machine-integers in the code, but that would make a
gap between the formalization and the generated code and would require
trusting that no overflows occur. A better approach would require
formalizing machine-integers and their properties and using them
within the formalization itself.

Compared to the prior pen-and-paper work, the computer assisted
approach significantly reduces the complexity of mathematical
arguments behind the proof and employs computing-machinery in doing
its best --- quickly enumerating and checking a large search
space. This enables formulation of a general framework for checking
various FC-families, without the need of employing human intellectual
resources in analyzing specificities of separate families. Compared to
the work of \v Zivkovi\' c and Vu\v ckovi\' c, apart from achieving
the highest level of trust possible, the significant contribution of
the formalization is the clear separation of mathematical background
and combinatorial search algorithms, not present in earlier
work. Also, separation of abstract properties of search algorithms and
technical details of their implementation significantly simplifies
reasoning about their correctness and brings them much closer to
classic mathematical audience, not inclined towards computer science.

This work represents a significant part in formally proving the
Frankl's conjecture for families $F$ such that $\card{\Union{F}} \le
11$, and $\card{\Union{F}} \le 12$ (already informally done by \v
Zivkovi\' c and Vu\v ckovi\' c \cite{frankl-zivkovic-vuckovic}) which
in the focus of our current and future work. We also plan to
investigate other FC-families (not necessarily uniform).


\end{document}